\documentclass[article]{amsart}
\date{April 14$^{\rm th}$  2015}
\usepackage{amssymb}
\usepackage{graphicx, amsmath}
\usepackage{amsfonts}
\usepackage{amssymb}
\usepackage{fancyhdr}
\usepackage{xcolor}

\newcommand{\proj}{P} 
\newcommand{\ef}{\xi} 

\newcommand{\N}{\mathbb{N}} 
\newcommand{\C}{\mathbb{C}} 
\newcommand{\Z}{\mathbb{Z}} 
\newcommand{\R}{\mathbb{R}} 

\newcommand{\tf}{\tilde{f}}

\newtheorem{proposition}{Proposition}[section]

\newtheorem{theorem}[proposition]{Theorem}
\newtheorem{lemma}[proposition]{Lemma}
\newtheorem{definition}[proposition]{Definition}
\newtheorem{remark}[proposition]{Remark}
\newtheorem{hypothesis}{Hypothesis}

\newcommand{\ri}{i}

\newcommand{\re}{\mathrm{e\,}}

\newcommand{\bx}{\mathbf{x}}
\newcommand{\id}{\mathbf{1}}
\newcommand{\style}{\displaystyle}
 \newcommand{\e}{\mathrm{e}}
\newcommand{\dd}{\mathrm{d}}
\newcommand{\scal}[1]{\langle #1 \rangle}
\newcommand{\ro}{\varrho}
\newcommand{\com}[1]{\left[#1\right]}
\newcommand{\norm}[1]{\left\lVert #1 \right\rVert}
\renewcommand{\Re}{\mathrm{Re}}
\newcommand{\bsigma}{\boldsymbol{\sigma}}
\newcommand{\bA}{\mathbf{A}}

\newcommand{\F}{\mathcal{F}}

\newcommand{\vp}{\varphi}

\title[Dynamical localization of  Dirac particles]{Dynamical localization of  Dirac particles
  in electromagnetic fields with dominating magnetic potentials}

\author{Jean-Marie Barbaroux}

\address{Jean-Marie Barbaroux\\
 Aix-Marseille Universit\'e, CNRS, CPT, UMR 7332, 13288 Marseille, France\\
 and Universit\'e de Toulon, CNRS, CPT, UMR 7332, 83957 La Garde, France.}
\email{barbarou@univ-tln.fr}
\author{Josef Mehringer}
\address{Josef Mehringer\\
Mathematisches Institut\\
Ludwig-Maximilians-Universit\"at\\
Theresienstra{\ss}e 39\\
D-80333 M\"unchen, Germany.}
\email{josef.mehringer@math.lmu.de}
\author{Edgardo Stockmeyer}
\address{ Edgardo Stockmeyer\\ Instituto de F\'\i sica\\
Pontificia Universidad Cat\'olica de Chile\\
Vicu\~na Mackenna 4860\\
 Santiago 7820436, Chile.}
\email{stock@fis.puc.cl}
\author{Amal Taarabt}
\address{Amal Taarabt\\ Instituto de F\'\i sica\\
Pontificia Universidad Cat\'olica de Chile\\
Vicu\~na Mackenna 4860\\
 Santiago 7820436, Chile.}
\email{ataarabt@fis.puc.cl}%
\begin{document}
\begin{abstract}
  We consider two-dimensional massless Dirac operators in a radially
  symmetric electromagnetic field. In this case the fields may be
  described by one-dimensional electric and magnetic potentials $V$
  and $A$.  We show dynamical localization in the regime when
  $\style\lim_{r\to\infty}|V|/|A |<1$, where dense point spectrum occurs.
\end{abstract}

\maketitle
\section{Introduction}
Graphene is a two dimensional material consisting of carbon atoms
arranged in a honeycomb lattice which was isolated in 2004
\cite{Novoselov2004}. Behind its remarkable properties such as Klein
tunneling and finite minimal conductivity \cite{katsnelson2006chiral}
stays the fact that at low excitations energies the dynamics of charge carriers is
described by the massless two-dimensional Dirac operator
\cite{castro2009electronic}. For technological devices based on graphene one
needs the ability to confine and control the mobility of charge carriers. However,
confining Dirac particles is not an easy task due the so-called Klein
effect, where particles are able to penetrate electric potential walls
\cite{katsnelson2006chiral} with very little reflexion index.  In \cite{giavaras2009magnetic} it was
argued that in  presence of rotational symmetric electric and
magnetic fields one could confine or deconfine Dirac particles by
manipulating the strength of the fields at infinity, i.e., far away
from the sample. Our main result is a dynamical statement on this
effect and a continuation of a recent work \cite{MS2014B} by
two of the present authors. Before presenting  the result let us explain this
phenomenon with more mathematical details.

Denote by $H$ the two-dimensional massless Dirac operator coupled to a
radially symmetric  field ${\bf E}=E\hat{r}$ on the plane and
a radially symmetric transversal magnetic field $B$. If
the fields are sufficiently regular $H$ is a self-adjoint operator densely defined in $L^2(\R^2,\C^2)$ acting as
\begin{equation}\label{op H}
 H = \bsigma\cdot(-i\nabla-\bA) + V,
\end{equation}
where $V:\R^2\to\R$ is the  electric potential satisfying
\begin{align*}
  V(\bx)=-\int_0^{|\bx|} E(s) ds\equiv V(r)
\end{align*}
(abusing  notation, we write $ V(r)$
to denote $V(\bx)$,  where $r=|\bx| \in [0,\infty)$ is the standard
radial variable). Here $\bsigma=(\sigma_1,\sigma_2)$ is a matrix-valued vector whose components are the
first two Pauli matrices
$$\sigma_1=\begin{pmatrix}0&1\\1&0\end{pmatrix}, \
\sigma_2=\begin{pmatrix}0&-i\\i&0 \end{pmatrix}.$$
The vector potential $\bA=(A_1,A_2):\R^2\to\R^2$ generates the magnetic field $B$, with
$B=\partial_1A_2-\partial_2 A_1$. We choose the rotational gauge, i.e., we set
\begin{align*}
  \bA(\bx):=\frac{1}{r} A(r) \begin{pmatrix}
-x_2\\x_1
\end{pmatrix},\quad\mbox{with}\quad A(r)=\frac1r \int_0^rB(s)s
\dd s.
\end{align*}
We note that, besides some local regularity requirements for $(V,A)$,
$H$ is essentially self-adjoint on $C_0^\infty(\R^2,\C^2)$  independently of the growth
rate of $A$ and $V$ (see \cite{Chernoff77}).

In this setting there exists a unitary transform  \cite{Thaller} (see
also \cite[Section 6]{KS})
\begin{equation}\label{U}
\mathcal{U}: L^2(\R^2,\C^2)\longrightarrow L^2(\R^+,\C^2)\otimes \ell^2(\mathbb{Z})\, ,
\end{equation}
such that the operator $H$ can be written as a direct sum of operators on the half-line
\begin{equation}\label{dir sum}
 \mathcal{U}H \mathcal{U}^*=\bigoplus_{j\in\Z} h_j,
\end{equation}
where
\begin{equation}\label{h_j}
 h_j=-i\sigma_2 \partial_r+\sigma_1(A-\tfrac{m_j}{r}) +V \
 \quad{\mbox{on}}\quad
L^2(\R^+,\C^2),
\end{equation}
with $m_j=j+\tfrac{1}{2}$ for $j\in\Z$. Clearly, the spectra of $H$
and $h_j$ are related through
\begin{align}
\sigma(H)=\overline{\bigcup_{j\in \Z} \sigma(h_j)},\quad \sigma_c(H)=\overline{\bigcup_{j\in \Z} \sigma_c(h_j)},\quad \mbox{and}\quad\sigma_{pp}(H)={\bigcup_{j\in \Z} \sigma_{pp}(h_j)}.
\end{align}
For the operators $h_j$ we have the following properties, assuming
sufficiently regular fields: If $A(r)\to \infty$ as
$r\to \infty$ and
\begin{align}
  \label{eq:9l}
  \overline{\lim _{r\to\infty} }\left|
    \frac{V(r)}{A(r)}\right|<1,
\end{align}
then the spectrum of $h_j$ is discrete for each $j\in \Z$ (see
\cite[Proposition 1]{MS2014} for the precise regularity
conditions). In contrast, if $V(r)\to \infty$ as
$r\to \infty$ and
\begin{align}
  \label{eq:9d}
  \overline{\lim _{r\to\infty} }\left|
    \frac{A(r)}{V(r)}\right|<1,
\end{align}
then the spectrum of $h_j$ equals the whole real line and it is purely
absolutely continuous \cite[Propsition 2]{KMSYamada}. This suggests
delocalized particles in the regime given by \eqref{eq:9d} and
confined particles in the one given by \eqref{eq:9l}. However, the
latter is not obvious since, for fields satisfying \eqref{eq:9l}, $H$
may have dense point spectrum (see \cite{miller1980quantum} and
\cite[Theorem 7.10]{Thaller} for the case when $B$ decays at infinity
and \cite{MS2014} for the case when $B$ is not decaying
at infinity). We recall that dense point spectrum may lead to
non-trivial dynamics. In fact, in this case, it is only known that the wave-packet
spreading is sub-ballistic \cite{Simon1990} (the result stated in
\cite{Simon1990} is for Laplace-type operators but can easily be
adapted for the Dirac case). Moreover, there are examples of systems
with pure point spectrum where
the spreading rate is arbitrarily close to the ballistic one
\cite{RJLS2,RJLS96} (see also \cite{BT99}).

Concerning dynamical results we know that particles in electromagnetic
fields satisfying \eqref{eq:9d} behave ballistically, i.e., for any
finite energy state $\psi\in L^2(\R^2,\C^2)$ and $\kappa>0$ one has
(see \cite{MS2014B})
\begin{align*}
  \frac1T \int_0^\infty \left\| |\bx|^{\kappa/2} \e^{-iHt}\psi \right\|^2 dt \sim T^\kappa,
  \quad \mbox{for large}\quad T>0.
\end{align*}
The main result of this work is that under condition
\eqref{eq:9l} the operator $H$ exhibits  dynamical localization, i.e,
for any $\kappa>0$, for any finite energy interval $I$, and for any
state $\psi\in L^2(\R^2,\C^2)$, with sufficient  regularity in the
angular variable (depending on $\kappa$; see \eqref{eq:133}), holds
\begin{align*}
  \sup_{t\ge 0} \left\||\bx|^{\kappa/2}\ \e^{-\ri t H} \proj_I(H) \psi\right\|^2 <\infty,
\end{align*}
where $P_I(H)$ is the spectral projection of $H$ onto $I$.

Let us now state our assumptions and result more precisely:
\begin{hypothesis}\label{A1}
$A,V\in {C}^1(\R^+,\R)$ and they satisfy
\begin{align}\label{con0}
& |A(r)| \to\infty \ \quad{as} \ \ r\to\infty,\\\label{con1}
  &\overline{\lim _{r\to\infty} }\left|
    \frac{V(r)}{A(r)}\right|<1,\\
\label{con2}
&\lim _{r\to\infty} \left|
    \frac{A'(r)}{A^2(r)}\right|=0.
\end{align}
\end{hypothesis}
Recall that $\mathcal{U}$ (see \eqref{U}) is the unitary map that decomposes $H$ in the
direct sum of the operators $h_j$. For a given $\psi\in L^2(\R^2,\C^2)$ we write
\begin{align}
  \label{eq:14}
  \mathcal{U}\psi=\mathop\oplus_{j\in\Z} \varphi_j,
  \quad\mbox{with}\quad \varphi_j\in L^2(\R^+,\C^2).
\end{align}

Our main result is the following theorem.
\begin{theorem}\label{dyn loc}
  Let $\kappa>0$, $I\subset\R$ be a bounded energy interval and let $\proj_I(H)$ be the spectral projection of $H$
  onto $I$. Assume that $A$ and $V$ fulfill Hypothesis \ref{A1} and let $\psi\in
  L^2(\R^2,\C^2)$ be a normalized state such that 
  $\mathcal{U}\psi=\style\mathop\oplus_{j\in\Z} \varphi_j$ satisfies
\begin{align}
  \label{eq:133}
  \sum_{j\in\Z} |j|^\kappa \norm{\varphi_j}^2<\infty.
\end{align}
Then we have
 \begin{equation}\label{eq:1}
  \sup_{t\ge 0} \norm{|\bx|^{\kappa/2}\ \e^{-\ri t H} \proj_I(H)\psi}^2 <\infty.
 \end{equation}
\end{theorem}
We note that the condition \eqref{eq:133} is related to regularity of
the initial state $\psi$ in the angular variable.  Indeed, let $r^{-1/2}
\tilde\psi$ with $\tilde\psi\in L^2(\R^+\times[0,2\pi), \C^2)$ be
equal to $\psi\in L^2(\R^2,\C^2)$ expressed in polar
coordinates. Then, \eqref{eq:133} follows if $\tilde\psi\in
H^{\kappa/2}([0,2\pi),L^2(\R^+))\oplus
H^{\kappa/2}([0,2\pi),L^2(\R^+))$. Here, for $\kappa>0$
  \begin{equation*} \style
   H^{\kappa/2}\left([0,2\pi),L^2(\R^+)\right):=\{u\in
   L^2([0,2\pi)\times\R^+),
   \sum_{\ell\in\Z}(1+|\ell|)^{\kappa}\norm{\hat
  u_\ell}^2_{L^2(\R^+)} <\infty\},
  \end{equation*}
  is the fractional Sobolev space on the torus \cite{Gr} and $\hat u_\ell$
  is the $\ell-\mathrm{th}$ Fourier coefficient of $u$ with respect to
  the variable $\theta$. To be more explicit, note that
  $\style\F\tilde\psi =\mathop\oplus_{j} \varphi_j$ where, for $g\in
  L^2([0,2\pi),\C^2)$, \begin{equation*}
  (\F g)(j):=\frac{1}{\sqrt{2\pi}}\int_0^{2\pi} M_\theta\
  \e^{-im_j\theta}g(\theta)\dd\theta,
  \quad\mbox{and}\quad M_\theta=
  \begin{pmatrix}\e^{i\theta/2}&0\\0&i\e^{-i\theta/2}\end{pmatrix}.
\end{equation*}
We notice that
\begin{equation*}
  (\F g)(j)=\begin{pmatrix}\hat g_1(j)\\ i\hat g_2(j+1) \end{pmatrix}.
\end{equation*}
The angular momentum operator $J:=-i\partial_\theta+\sigma_3/2$ satisfies $\F
J\F^*=m_j$. Assuming for simplicity that $\kappa=2n$ with $n\in \N$, we have
\begin{align*}
  \big(\sum_{j\in\Z}|m_j|^{2n}\norm{\varphi_j}^2\big)^{1/2}&
  =\norm{J^{n} \tilde\psi}
  =\norm{(-i\partial_\theta+\frac{\sigma_3}{2})^{n} \tilde\psi}
  \\&\quad\le\sum_{k=0}^{n} {n \choose k}\norm{(-i\partial_\theta)^k\tilde\psi}
  \le2^n\norm{\tilde\psi}_{H^n([0,2\pi),L^2(\R^+))^2}.
\end{align*}
We can also derive a sufficient condition on $\psi$ for \eqref{eq:133} to hold. To avoid unnecessary complications, we stick to the case of even values of $\kappa$. Since $\partial_\theta = -x_2 \partial x_1 + x_1\partial x_2$,
we have
\begin{equation*}
   \norm{(\partial_\theta)^k \tilde\psi}^2
 = \norm{(-x_2 \partial_{x_1} + x_1\partial_{x_2})^k\psi}^2
\leq C_k\style\sum_{|\alpha|\leq k}\norm{(1 + |\bx|^2)^\frac{k}{2} \partial^\alpha_\bx \psi}^2,
\end{equation*}
where for the multi-index $\alpha=(\alpha_1, \alpha_2)$ we used the
notation $\partial^\alpha_\bx = \partial_{x_1}^{\alpha_1} \partial_{x_2}^{\alpha_2}$.
Hence, in Cartesian coordinates, condition \eqref{eq:133} holds if $\psi$ belongs to the weighted Sobolev space $H^{\kappa/2}_{\kappa/2}(\R^2,\C^2) := \{ g \in L^2(\R^2,\C^2)\ | \
  (1+|\bx|^2)^{\frac{\kappa}{4}} g(\bx) \in H^{\kappa/2}(\R^2,\C^2) \}$.
\section{Proof of the main result}

The strategy of the proof is as follows. We first rewrite the moment
in \eqref{eq:1} using the decomposition of $H$ as a direct sum of
$h_j$'s and the representation $\psi\simeq\mathop\oplus
\varphi_j$ (see \eqref{eq:2}). The main idea consists in taking
advantage of the discrete spectrum of the operator $h_j$ and the
exponential decay of its eigenfunctions to compensate the growth of the
moment. However, the decay is not uniform in $j$, it turns out (see
\eqref{eq:12} in Lemma~\ref{ex-decay}) that the exponential decay
estimate can only be derived outside an interval which grows with
$j$.  This is due to the fact that the term $A-m_j/r$ in $h_j$ can be
controlled by a fraction of $A$ only outside a ball of radius $r_j$
with $r_j |A(r_j)| \sim |m_j|$. Hence, we split the space
$\R^+\times\Z$ into two regions through the function $f$, given in
Definition~\ref{theta}. In the first region, where the values of $r
|A(r)|$ are sufficiently large compared to those of $j$, we
control the growth of the moment in $r$ with the exponential decay of
the eigenfunctions obtained in Section~\ref{use2} by an Agmon-type argument. We
also need to control the number of these eigenfunctions in the
spectral region we consider. This latter bound is derived in Section~\ref{use1} by using arguments due to
Bargmann. In the second region where the values of $j$ are large
compared to $r |A(r)|$, the regularity assumption \eqref{eq:133} in
the angular variable for our initial state yields a decay in the
variable $j$ which is used to control the growth of $r^\kappa
\lesssim j^\kappa$.

We now turn to the detailed proof of Theorem~\ref{dyn loc}.

Since $\mathcal{U}\psi=\oplus_{j}\varphi_j$ we have
\begin{equation}\label{eq:2}
 \norm{|\bx|^{\kappa/2}\ \e^{-\ri t H} \proj_I(H)\psi}^2 =\sum_{j\in\Z} \norm{ r^{\kappa/2} \re^{-\ri t h_j} \proj_I(h_j) \vp_j}^2.
\end{equation}
We  note that  it is enough to consider  only the sum for $|j| > J_0$, for
sufficiently large $J_0>1$. Indeed, let $N_j$ be the number of
eigenvalues of $h_j$ in the interval $I$. Let  $\ef_j^{(k)}$, $k=1,2,\dots,
 N_j$ denote the corresponding  eigenfunctions of
$h_j$. Then, expanding $ \proj_I(h_j) \vp_j$ in terms of the $\ef_j^{(k)}$, we have
\begin{align*}
 \sum_{|j|\le J_0} \norm{ r^{\kappa/2} \re^{-\ri t h_j} \proj_I(h_j) \vp_j}^2
  \le \sum_{|j|\le J_0} \sum_{k=1}^{N_j} \norm{r^{\kappa/2}\ef_j^{(k)}}^2<\infty.
\end{align*}
where the last inequality  holds in view of Remark \ref{ladyd} below.
To estimate the right hand side of \eqref{eq:2} we split the
integration $\style\sum_j \int_0^\infty$ into two regions characterized by
\begin{equation}\label{eq:def-rj}
 r_j\equiv r_j(\delta_0) := \sup\{ r \in \R^+ \ |\ |m_j| \geq \delta_0 r |A(r)|\},
\end{equation}
for some $\delta_0\in(0,1)$ which is chosen in the proof of
Lemma \ref{ex-decay} below.
\begin{remark}\label{rjr}
  Note that since $A$ is continuous and $A(r)\to\infty$ as
  $r\to\infty$ we have that $r_j<\infty$ and that $r_j\to\infty$ as $|j|\to
  \infty$.  Moreover, the supremum in \eqref{eq:def-rj} is attained
  and hence
\begin{align}
  \label{eq:3}
  |m_j|=\delta_0 r_j |A(r_j)|.
\end{align}
Moreover, we note that
\begin{align}
  \label{eq:31}
  |m_j|\le\delta_0 r|A(r)|,\quad r\ge r_j.
\end{align}
\end{remark}
\begin{definition}\label{theta}
Let $\theta\in C^\infty(\R^+,[0,1])$ with $\theta(r)=0$ for $r<1$ and
$\theta(r)=1$ for $r>2$. We set $f(r,j)= f_j(r) :=\theta(r/3 r_j) $
and $f_j^c:=1-f_j$.
\end{definition}

\begin{proof}[Proof of Theorem~\ref{dyn loc}.]
We have
\begin{equation}\label{eq:split}
\begin{split}
 & \sum_{j\in\Z} \norm{ r^{\kappa/2} \re^{-\ri t h_j} \proj_I(h_j) \vp_j}^2 \\
 &  \leq 2 \sum_{j\in \Z} \left( \norm{f_j r^{\kappa/2} \re^{-\ri t h_j} \proj_I(h_j) \vp_j}^2
  + \norm{ f_j^c r^{\kappa/2} \re^{-\ri t h_j} \proj_I(h_j) \vp_j}^2  \right).
\end{split}
\end{equation}
We first estimate the second term of the right hand side of
\eqref{eq:split} using the regularity in the angular variable for the initial state as given by \eqref{eq:133}. In what follows we pick $J_0>1$ so large that $|A(r)|>1$ for
all $r>r_j$ and $|j|>J_0$. In particular, we have that
\begin{align}
  \label{eq:4}
  r_j<|m_j|/\delta_0.
\end{align}
Using this and the support properties of $f_j^c$, we have
\begin{equation}\label{eq:large-mj}
\begin{split}
 &  \sum_{|j|> J_0} \norm{ f_j^c r^{\kappa/2} \re^{-\ri t h_j} \proj_I(h_j) \vp_j}^2 \\
 & \leq \sum_{|j|> J_0} \norm{ f_j^c (6r_j)^{\kappa/2}  \re^{-\ri t h_j} \proj_I(h_j) \vp_j}^2 \\
 & \leq \sum_{|j|> J_0} (6|m_j|/\delta_0)^{\kappa}
\norm{ \re^{-i t h_j}\proj_I(h_j)}^2 \, \norm{\vp_j}^2  <\infty\, ,\\
\end{split}
\end{equation}
where we used \eqref{eq:133} in the last bound.

We now estimate the first term in the right hand side of
\eqref{eq:split}. For  $|j|>J_0$ we compute
\begin{align*}
  \norm{f_jr^{\kappa/2}\proj_I(h_j)\e^{i t h_j} \varphi_j}^2&\le
  \sup_{\|\phi\|=1}
\left(\sum_{k=1}^{N_j} \norm{f_jr^{\kappa/2} \scal{\ef_j^{(k)},\phi}
  \ef_j^{(k)}}\right)^2\\
&\le
   \sum_{k=1}^{N_j} \norm{f_jr^{\kappa/2}
  \ef_j^{(k)}}^2  \,\sup_{\|\phi\|=1} \sum_{l=1}^{N_j}
  |\scal{\ef_j^{(l)},\phi}|^2\\
&= \sum_{k=1}^{N_j} \norm{f_jr^{\kappa/2}
  \ef_j^{(k)}}^2.
\end{align*}
Consider the function $\tf_j$ defined at the beginning of  Section~\ref{use2} below. Since
$f_j=f_j\tf_j$ we have, choosing also $J_0>J_2$ (see Lemma \ref{ex-decay})
\begin{align*}
  \norm{f_jr^{\kappa/2}
  \ef_j^{(k)}}\le \norm{f_jr^{\kappa/2}\e^{-\gamma \ro}}
  \,\norm{\e^{\gamma \ro}\tf_j\ef_j^{(k)}}\le
  \norm{f_jr^{\kappa/2}\e^{-\gamma \ro}} \frac{C}{r_j} \e^{\gamma \ro(2r_j)},
\end{align*}
where $\rho(r)=\int_0^r|A(s)|\dd s$ is the exponential weight
defined in  Lemma~\ref{ex-decay}.
Since $r^{\kappa/2}\e^{-\gamma \ro}$ decays monotonically at infinity,
we may choose $J_0>1$ to be so large that the supremum of
$f_jr^{\kappa/2}\e^{-\gamma \ro}$ is bounded above by
$(3r_j)^{\kappa/2}\e^{-\gamma \ro(3r_j)}$. Hence
  \begin{align*}
    \norm{f_jr^{\kappa/2}
  \ef_j^{(k)}}\le \frac{C}{r_j}  (3r_j)^{\kappa/2}\e^{-\gamma (\ro(3r_j)-\ro(2r_j))}.
  \end{align*}
Note that due to \eqref{eq:31}
\begin{align*}
  \ro(3r_j)-\ro(2r_j)=\int_{2r_j}^{3r_j} r|A(r)| \frac{\dd r}{r}\ge
  \frac{|m_j|}{\delta_0} \ln(\tfrac{3}{2} )>\frac{|m_j|}{3\delta_0},
\end{align*}
for $|j|$ large enough. Thus, using that $r_j\le |m_j|$ and
Lemma \ref{bargmann2}, we get for $|j|>J_2$
\begin{align*}
  \norm{f_jr^{\kappa/2}\proj_I(h_j)\e^{i t h_j} \varphi_j}^2&\le C^2 \sum_{k=1}^{N_j}
  \left(\e^{-\gamma \frac{|m_j|}{3\delta_{0}}}
    (3|m_j|)^{\kappa/2}\right)^2\\
 &\le 3^\kappa{C^2 C_{I}} |m_j|^{\kappa+1}   \ln\!|m_j|\,
\e^{-\gamma \frac{2|m_j|}{3\delta_0}}.
\end{align*}
Since the last bound is summable for $|j|=|m_j-1/2|>J_0$ we get the expected result.
\end{proof}
\section{Estimate on the number of eigenvalues of $h_j$}\label{use1}
Let $T$ be a self-adjoint operator on a Hilbert space $\mathcal H$
with purely discrete spectrum.
We set for an interval $I \subset \R$
\begin{equation*}
N_I(T) := \dim P_I(T) \mathcal{H}  ,
\end{equation*}
i.e. $N_I(T)$ denotes the number of eigenvalues of $T$ in $I$
counted with multiplicity.
\begin{lemma}[Bound on the number of eigenvalues for $h_j$]\label{bargmann2}
There is a $J_1>1$ such that for any $E>0$ there is a constant $C_E>0$
so that
\begin{equation}
 N_{[-E,\, E]}(h_j) \le C_E |m_j| \ln\!|m_j| , \quad
\mathrm{for} \ |j| \ge J_1.
\end{equation}
\end{lemma}
\begin{proof}
We first note that
\begin{align*}
h_j =  - \ri\sigma_2 \partial_r + \sigma_1 \left(A(r) - \tfrac{m_j}{r}\right) + V(r)
\end{align*}
is essentially
self-adjoint on $C_0^\infty(\R^+,\C^2)$. In addition, we
obtain by the spectral theorem
\begin{equation*}
 N_{[-E,\, E]}(h_j) =  N_{[0,\, E^2]}(h_j^2).
\end{equation*}
In the sense of quadratic forms on $C_0^\infty(\R^+,\C^2)$
we obtain for any $\epsilon \in (0,1)$ the estimate
\begin{align*}
h_j^2 \ge &
(1-\epsilon) \big[  -\ri\sigma_2 \partial_r - \sigma_1\tfrac{m_j}{r}
+ \sigma_1 A(r) \big]^2 + \big(1- \tfrac{1}{\epsilon}\big) V^2(r)\\
= &
(1-\epsilon) \big[
\big(-\ri\sigma_2 \partial_r - \sigma_1\tfrac{m_j}{r} \big)^2
+A^2(r) -\tfrac{1}{\epsilon}V^2(r)
-\sigma_3 A'(r) - \tfrac{2m_j}{r}A(r) \big].
\end{align*}
Let
$\delta:=(1-\epsilon)/2$. Due to \eqref{con1} and \eqref{con2} we have
that $\delta A^2-\sigma_3A'$ and $\epsilon A^2-V^2/\epsilon$ are
positive outside a large ball $B$ for $\epsilon\in(0,1)$ sufficiently
close to $1$. Let
$C:=\|V^2/\epsilon+\sigma_3A'\|_{L^\infty(B)} $. Then we find
\begin{align*}
h_j^2 \ge (1-\epsilon) \Big[
\big(-\ri\sigma_2 \partial_r - \sigma_1\tfrac{m_j}{r} \big)^2
+\delta A^2(r) -C  - \tfrac{2|m_j|}{r}|A(r)| \Big].
\end{align*}
We write
\begin{align*}
  h_j^2 - E^2\ge
 (1-\epsilon) \Big[
\big(-\ri\sigma_2 \partial_r - \sigma_1\tfrac{m_j}{r} \big)^2
+ W_j \Big],
\end{align*}
where
\begin{align*}
  W_j(r)=\delta A^2(r) -C  -(1-\epsilon)^{-1} E^2- \tfrac{2|m_j|}{r}|A(r)|.
\end{align*}
Let $R_j:=r_j(\delta/4)$, where $r_j(\cdot)$ is defined in
\eqref{eq:def-rj}. This yields,
$|m_j|\le \frac\delta4 r |A(r)| $ for all $r>R_j$.
Moreover, we may pick $J_1$ so large that (recall that $R_j\to\infty$
as $|j|\to \infty$)
\begin{align*}
  \tfrac{\delta}{2} A^2(r) -C  -(1-\epsilon)^{-1} E^2>0, \quad
  \mbox{for all} \quad r>R_j,\, |j|>J_1.
\end{align*}
Thus , $W_j \id_{(R_j,\infty)}\ge 0$ and
\begin{align}
  \label{eq:5}
   h_j^2 - E^2\ge (1-\epsilon) \Big[
\big(-\ri\sigma_2 \partial_r - \sigma_1\tfrac{m_j}{r} \big)^2
+ W_j \id_{(0,R_j]}\Big].
\end{align}
Define
\begin{align}
  \label{eq:6}
  D_j:=\{r\in(0,R_j)\ | \ |m_j|\ge \tfrac{\delta}{4} r|A(r)|\}.
\end{align}
Note that if $r\in (0,R_j)\cap (\R^+\setminus D_j)$ then
$\frac{\delta}{2}A(r)^2-\frac{2|m_j|}{r} |A(r)|\ge 0$. Hence we have
\begin{align}
  \label{eq:7}
    h_j^2 - E^2\ge  (1-\epsilon) \Big[
\big(-\ri\sigma_2 \partial_r - \sigma_1\tfrac{m_j}{r} \big)^2
+ W_j^< \Big],
\end{align}
where
\begin{align}
  \label{eq:8}
   W_j^<(r):= (\delta A^2(r) - \tfrac{2|m_j|}{r}|A(r)|)\id_{D_j} -
(C  +(1-\epsilon)^{-1} E^2)\id_{(0,R_j]}.
\end{align}
An application of the min-max principle leads to
\begin{align*}
N_{[0, E^2]}\big(h_j^2\big) &=
N_{(-\infty, 0]} \big( h_j^2 -E^2\big) \\ & \le
N_{(-\infty, 0]} \big(
\big(-\ri\sigma_2 \partial_r - \sigma_1\tfrac{m_j}{r} \big)^2
+W_j^<\big).
\end{align*}
A direct computation shows that
\begin{align*}
\big(-\ri\sigma_2 \partial_r - \sigma_1\tfrac{m_j}{r} \big)^2
&=-\partial_r^2 + \tfrac{1}{r^2}
m_j (m_j-\sigma_3)\\
&=
\begin{pmatrix}
-\partial_r^2 + \tfrac{1}{r^2}\,
m_j (m_j-1) & 0\\
0 &-\partial_r^2 + \tfrac{1}{r^2}\,
m_j(m_j+1)
\end{pmatrix}.
\end{align*}
Note that $m_j(m_j \pm 1) > 0$ for $|j| >J_1$. Using the
generalized Bargmann estimate \cite{B1952} (see also \cite[Theorem XIII.9]{RS4}) we obtain for $|m_j| >1/2$
\begin{align*}
N_{(-\infty, 0]} \big(
-\partial_r^2 + \tfrac{1}{r^2}\, m_j (m_j&-1)
+ W_j^< \big) \\ &\le
\begin{cases}
\tfrac{1}{2m_j-1} \int_0^\infty r| W_j^< (r)| \dd r
& \mathrm{if} \ m_j > \tfrac{1}{2} \\[0.1cm]
\tfrac{1}{2|m_j|+1} \int_0^\infty r| W_j^< (r)| \dd r
& \mathrm{if} \ m_j < - \tfrac{1}{2}
\end{cases}
\end{align*}
and
\begin{align*}
N_{(-\infty, 0]} \big(
-\partial_r^2 + \tfrac{1}{r^2}\, m_j (m_j&+1)
+  W_j^< \big) \\ &\le
\begin{cases}
\tfrac{1}{2m_j+1} \int_0^\infty r| W_j^< (r)| \dd r
& \mathrm{if} \ m_j > \tfrac{1}{2} \\[0.1cm]
\tfrac{1}{2|m_j|-1} \int_0^\infty r| W_j^< (r)| \dd r
& \mathrm{if} \ m_j < - \tfrac{1}{2}
\end{cases},
\end{align*}
and therefore
\begin{align}\label{isabel1}
N_{(-\infty, 0]} \big(
\big(-\ri\sigma_2 \partial_r - \sigma_1\tfrac{m_j}{r} \big)^2
+ W_j^< \big)  \le
\frac{1}{|m_j|- 1/2} \int_0^\infty r| W_j^<(r)| \dd r.
\end{align}
Now we estimate using the definition of $D_j$
\begin{equation}
\label{isabel2}
\begin{split}
\int_0^\infty r| W_j^<(r)|  \dd r  &\le \frac{(C
  +(1-\epsilon)^{-1} E^2)R_j^2}{2} +
\int_{D_j}
r\big|\delta A^2(r) -
\tfrac{2|m_j|}{r}|A(r)| \big| \dd r\\
  &\le \frac{(C
  +(1-\epsilon)^{-1} E^2)R_j^2}{2} +
\int_{D_j} 2|m_j||A(r)| \dd r.
\end{split}
\end{equation}
Furthermore,
\begin{equation}
\label{isabel3}
\begin{split}
\int_{D_j} 2|m_j||A(r)| \dd r&=  \int_{D_j\cap(0,1)}
  2|m_j||A(r)| \dd r
+  \int_{D_j\cap(1,\infty)} 2|m_j||A(r)| \dd r\\
&\le 2|m_j|\norm{A}_{L^\infty[0,1]}+\tfrac{8m_j^2}{\delta} \ln(R_j).
\end{split}
\end{equation}
Note that in view of Remark \ref{rjr}, and the fact that $|A(r)|$
grows at infinity, we have for sufficiently large $J_1$
\begin{align*}
  R_j=\frac{4|m_j|}{\delta |A(R_j)|}\le \frac{4|m_j|}{\delta},\quad |j|>J_1.
\end{align*}
This together with \eqref{isabel1},  \eqref{isabel2}, and
\eqref{isabel3} yields the result.
\end{proof}
\section{Exponential decay of eigenfunctions of $h_j$}\label{use2}
Let  $r_j\equiv r_j(\delta_0)$ be   given  as in \eqref{eq:def-rj}. We
note that $\delta_0\in(0,1)$ will be fixed throughout the proof of the
next lemma. For the function $\theta$ as defined in Definition~\ref{theta}, we set $\tf(r,j)= \tf_j(r) :=\theta(r/ r_j)$.
\begin{lemma}\label{ex-decay}
There exist  $\gamma>0$ and  $J_2>1$ such that for all $|j|>J_2$ the following holds:
Let $\ef_j\in L^2(\R^+,\C^2)$ be a normalized eigenfunction of $h_j$ with energy $E\in I$. Then,
for some $C>0$ (independent of $j$),
\begin{align}
  \label{eq:12}
  \norm{A\e^{\gamma\ro}\tf_j\ef_j}\le \frac{C}{r_j} \e^{\gamma\ro(2r_j)}
\end{align}
where  for $r\geq 0$, $\style\ro(r):=\int_0^r |A(s)|\dd s$.
\end{lemma}
\begin{remark}
\label{ladyd}
It is clear from the proof that the exponential decay of the
eigenfunctions of $h_j$ remains true for $|j|\le J_2$, however, in
this case we get a different constant in front of the exponential.
\end{remark}
\begin{remark}\label{domain}
  Throughout the proof of Lemma \ref{ex-decay} we use that $ h_j$ and
  $k_j:=h_j-V$ are essentially self-adjoint operators on
  $C_0^\infty(\R^+,\C^2)$ (see for instance \cite{MS2014B} and
  references therein). Moreover, we also use that
  $V$ is a perturbation with respect to the magnetic Dirac operator
  $k_j$ in the sense that   $\mathcal D(V)\supset C_0^\infty(\R^+,\C^2)$ and
  there exists $C$ such that
\begin{equation*}
\norm{V \vp}\le C( \norm{k_j \vp}+\norm{\vp}) \ \quad{for \ all} \ \vp\in C_0^\infty(\R^+,\C^2).
\end{equation*}
Indeed, let $\vp\in C_0^\infty(\R^+,\C^2)$ and $\chi_R $ be a smooth characteristic function of
a ball of radius $R>0$. We set $\chi_R^c=1-\chi_R$ and $V=V^<+V^>$ where
$V^<:=V \chi_R$ and $V^>:=V\chi_R^c$. We thus have
\begin{align*}
\norm{V\vp}
&\le\norm{V^< \vp}+\norm{V^>\vp}\le C_R\norm{\vp}+\norm{A\chi_R^c\vp},
\end{align*}
since $V^<$ is bounded and $V$ is dominated by $A$ at infinity according to assumption \eqref{con1}.
Moreover, for $R$ large enough, we use \eqref{con2} and the identity $k_j^2 = -\partial_r^2 + A_j^2 -\sigma_3 A'_j$ to write
\begin{align*}
 \tfrac{1}{2}\norm{A\chi_R^c\vp}
\le\norm{k_j\chi_R^c\vp}
&\le\norm{(\nabla\chi_R^c)\vp}+\norm{\chi_R^c k_j\vp}
\le\norm{\nabla\chi_R^c}\norm{\vp}+\norm{k_j\vp}.
\end{align*}
\end{remark}
\begin{proof}[Proof of Lemma \ref{ex-decay}]
In order to derive the Agmon-type estimates, we follow \cite{KS}.
We set
\begin{equation}\label{A_j} A_j:=(A-\frac{m_j}{r}),\end{equation}
and we notice that $|A_j|\ge (1-\delta_0)|A|$ on the support of $\tf_{j}$.
Let $\ef_j$ be a normalized eigenfunction of $h_j$ associated to an energy $E$.
We define
\begin{equation}\label{rho}
g_j:= \e^{\gamma\ro_\epsilon}\tf_j\ef_j,
\end{equation}
where $\gamma\in(0,1)$ and $\ro_\epsilon=\dfrac{\ro}{1+\epsilon\ro}$
such that $\style\ro(r)=\int_0^r |A(s)|\dd s$. Note that $\ro_\epsilon$
is bounded and differentiable.
Consider the operator
\begin{equation}\label{k_j}
k_j=-i\sigma_2\partial_r+\sigma_1 A_j=h_j-V,
\end{equation}
and we define
\begin{equation}\label{Q}
Q_j:=\Re\scal{k_j\e^{\gamma\ro_\epsilon} g_j,k_j\e^{-\gamma\ro_\epsilon}g_j}.
\end{equation}
The task is to obtain bounds for $Q_j$ that will allow us to bound $g$ uniformly in $\epsilon$.


{\bf Lower bound}.
Notice that
\begin{equation}
 \com{k_j,\e^{\gamma\ro_\epsilon}}=-i\sigma_2 \gamma\ro'_\epsilon\e^{\gamma\ro_\epsilon},
\end{equation}
so that we rewrite
\begin{align*}
  Q_j
  &=\Re\scal{\e^{-\gamma\ro_\epsilon} k_j\e^{\gamma\ro_\epsilon} g_j,\e^{\gamma\ro_\epsilon} k_j\e^{-\gamma\ro_\epsilon}g_j}\notag\\
  &=\Re\scal{(k_j-i\gamma\ro_\epsilon'\sigma_2) g_j,(k_j+i\gamma\ro_\epsilon'\sigma_2)g_j}\notag=\norm{k_j g_j}^2 - \gamma^2\norm{\ro_\epsilon'g_j}^2.
\end{align*}
Moreover, we have
\begin{equation}\label{k_j 2}
k_j^2=-\partial_r^2+A_j^2-\sigma_3 A_j'.
\end{equation}
In view of Remark \ref{rjr} and \eqref{con2} for any $\tilde\epsilon>0$ there exists $J_{\tilde\epsilon}>0$ such that for all $|j|>J_{\tilde\epsilon}$ one has
\begin{equation*}
 \scal{g_j, A_j' g_j}\le \tilde\epsilon\scal{g_j,A^2 g_j}
\end{equation*}
and therefore
\begin{equation}\label{A_j > A}
\scal{g_j,(A_j^2-\sigma_3A_j') g_j}\ge((1-\delta_0)^2-\tilde\epsilon)\scal{g_j,A^2g_j}.
\end{equation}
Now we drop the term $-\partial_r^2$ of \eqref{k_j 2}. This together with \eqref{A_j > A} yields
\begin{equation}\label{lower bnd}
\begin{split}
  Q_j&\ge((1-\delta_0)^2-\tilde\epsilon)\norm{Ag_j}^2-\gamma^2
  \norm{\ro_\epsilon'g_j}^2\\&\ge ((1-\delta_0)^2-\tilde\epsilon)\norm{Ag_j}^2-\gamma^2\norm{\ro'g_j}^2
=((1-\delta_0)^2-\tilde\epsilon-\gamma^2)
\norm{Ag_j}^2.
\end{split}
\end{equation}
{\bf Upper bound}.
We rewrite
\begin{align}
Q_j
&=\Re\scal{k_j\e^{\gamma\ro_\epsilon}g_j,\tf_{j}(E-V)\ef_j}
+\Re\scal{k_j\e^{\gamma\ro_\epsilon}g_j,-i\sigma_2\tf_{j}' \ef_j}\notag\\
&=\Re\scal{\e^{\gamma\ro_\epsilon}g_j,\tf_{j}(E-V)^2\ef_j}
+\Re\scal{\e^{\gamma\ro_\epsilon}g_j,\com{k_j,\tf_{j}(E-V)}\ef_j}\notag\\
&\quad +\Re\scal{\e^{\gamma\ro_\epsilon}g_j,-i\sigma_2\tf'_{j}(E-V)\ef_j}
+ \Re\scal{\e^{\gamma\ro_\epsilon}g_j,\com{k_j,-i\sigma_2 \tf'_{j}}\ef_j}\notag\\
&=\Re\scal{\e^{\gamma\ro_\epsilon}g_j,\tf_{j}(E-V)^2\ef_j}
+\Re\scal{\e^{\gamma\ro_\epsilon}g_j,\com{k_j,-i\sigma_2 \tf'_{j}}\ef_j},\notag
\end{align}
since
\begin{align*}
 \Re\scal{\e^{\gamma\ro_\epsilon}g_j,-i\sigma_2\tf'_{j}(E-V)\ef_j}= \Re\scal{\e^{\gamma\ro_\epsilon}g_j,\com{V,k_j}\tf_{j}\ef_j}=0.
\end{align*}
Furthermore, we use
\begin{align*}
\Re\scal{\e^{\gamma\ro_\epsilon}g_j,\tf_{j}(E-V)^2\ef_j}
=\norm{(E-V)\ g_j}^2.
\end{align*}
In addition, we find some $C>0$ such that
\begin{align*}
|\scal{\e^{\gamma\ro_\epsilon}g_j,\com{k_j,-i\sigma_2 \tf'_{j}}\ef_j}|
&=|\scal{\e^{\gamma\ro_\epsilon}g_j,(-\tf_{j}''+2\sigma_3A_j\tf_{j}')
  \ef_j}|\\
&\le C \frac{\e^{\gamma\ro(2r_j)}}{r_j}(\tfrac{1}{r_j}\norm{g_j}+\norm{A_j
g_j})\\
&\le C \frac{\e^{\gamma\ro(2r_j)}}{r_j}(\norm{g_j}+(1+\delta_0)\norm{A
g_j}),
\end{align*}
where in the last inequality we use that $r_j>1$ (for sufficiently
large $|j|$) and \eqref{A_j} together with the support properties of $\tf_j$.
We thus get
\begin{equation}\label{upper bnd}
\begin{split}
Q_j
&\le\norm{(E-V)\
    g_j}^2+ C \frac{\e^{\gamma\ro(2r_j)}}{r_j}(\norm{g_j}+(1+\delta_0)\norm{A g_j}).
\end{split}
\end{equation}
Then, combining \eqref{lower bnd} and \eqref{upper bnd} we get for $|j|>J_{\tilde{\epsilon}}$
\begin{equation}\label{eqqqq}
 \scal{g_j,(((1-\delta_0)^2-\tilde\epsilon-\gamma^2)A^2-(E-V)^2 )g_j}\le  C \frac{\e^{\gamma\ro(2r_j)}}{r_j}(\norm{g_j}+(1+\delta_0)\norm{A g_j}).
\end{equation}
According to \eqref{con1} and Remark \ref{rjr} we may pick
$\delta_0,\tilde\epsilon$ and $\gamma$ so small that there are constants
$J_{\delta_0,\tilde\epsilon,\gamma}, c_{\delta_0,\tilde\epsilon,\gamma}>0$
such that $|A|>1$ on the support of $\tf_j$ and, for all $|j|>J_{\delta_0,\tilde\epsilon,\gamma}$,
\begin{align}
  \label{eq:11}
  \scal{g_j,[((1-\delta_0)^2-\tilde\epsilon-\gamma^2)A^2-(E-V)^2 ]g_j}
\ge c_{\delta_0,\tilde\epsilon,\gamma}\norm{Ag_j}^2.
\end{align}
This together with \eqref{eqqqq} yields
\begin{align}\label{david}
 \norm{Ag_j}\le \frac{C}{c_{\delta_0,\tilde\epsilon,\gamma}}
 \frac{\e^{\gamma\ro(2r_j)}}{ r_j}(
{\norm{g_j}}/{\norm{Ag_j}}+(1+\delta_0))\le \frac{C}{c_{\delta_0,\tilde\epsilon,\gamma}}
 \frac{\e^{\gamma\ro(2r_j)}}{ r_j}(2+\delta_0).
\end{align}
The claim follows using the theorem of monotonic convergence for the
limit $\epsilon\to0$ of \eqref{david}.
\end{proof}
{\bf Acknowledgments.}  The authors want to thank Rafael Benguria for
useful discussions and remarks. J.-M.B has been supported by the project SQFT  ANR-12-JS01-0008-01. J.M. has been supported by SFB-TR12 ``Symmetries and
Universality in Mesoscopic Systems" of the DFG. E.S. has been supported by Fondecyt (Chile)
project 1141008 and Iniciativa Cient\'ifica Milenio (Chile) through
the Millenium Nucleus RC120002 ``F\'isica Matem\'atica''. A.T. has been supported by
the Millenium Nucleus RC120002 ``F\'isica Matem\'atica''.
\bibliographystyle{plain}
\bibliography{library}
\end{document}